\pdfoutput=1
\documentclass[11pt,letter]{article}
\usepackage{algorithmic}
\usepackage[ruled]{algorithm}
\usepackage{amsmath}
\usepackage{amssymb}
\usepackage{amsthm}
\usepackage{cite}
\usepackage{fullpage}
\usepackage{listings}
\usepackage{graphicx}
\usepackage{palatino}
\usepackage{url}

\long\def\omit#1{\relax}

\newtheorem{theorem}{Theorem}[section]
\newtheorem{lemma}[theorem]{Lemma}
\newtheorem{definition}[theorem]{Definition}

\begin{document}

\lstset{language=Python}

\title{Optimal Angular Resolution for Face-Symmetric Drawings}

\author{David Eppstein and Kevin A. Wortman \\
Department of Computer Science \\
Universitiy of California, Irvine \\
\emph{\{eppstein, kwortman\}@ics.uci.edu}}

\maketitle

\begin{abstract}
Let $G$ be a graph that may be drawn in the plane in such a way that all internal faces are centrally symmetric convex polygons. We show how to find a drawing of this type that maximizes the angular resolution of the drawing, the minimum angle between any two incident edges, in polynomial time, by reducing the problem to one of finding parametric shortest paths in an auxiliary graph.  The running time is at most $O(t^3)$, where $t$ is a parameter of the input graph that is at most $O(n)$ but is more typically proportional to $\sqrt{n}$.
\end{abstract}

\section{Introduction}
Angular resolution, the minimum angle between any two edges at the
same vertex, has been recognized as an important aesthetic criterion
in graph drawing since its introduction by Malitz and Papakostas in
1992~\cite{MalPap-STOC-92}. Much past work on angular resolution in
graph drawing has focused on bounding the resolution by some function
of the vertex degree or other related quantities, rather than on exact optimization; however, recently,
Eppstein and Carlson~\cite{EppCar-GD-06} showed that, when drawing
trees in such a way that all faces form (infinite) convex polygons,
the optimal angular resolution may be found by a simple linear time
algorithm. The resulting drawings have the convenient property that
the lengths of the tree edges may be chosen arbitrarily, keeping fixed
the angles selected by the optimization algorithm, and no crossing can
occur; therefore, one may choose edge lengths either to achieve other
aesthetic goals such as good vertex spacing or to convey additional
information about the tree.

In this paper, we consider similar problems of calculating edge angles
with optimal angular resolution for a more complicated class of graph
drawings: planar straight line drawings in which each internal face of
the drawing is a centrally-symmetric convex polygon. In previous
work~\cite{Epp-GD-04}, we investigated drawings of this type, which we called \emph{face-symmetric drawings}.  We characterized them as the duals of weak pseudoline arrangements in the
plane, and described an algorithm that finds a face-symmetric drawing
(if one exists) in linear time, based on an SPQR-tree decomposition of
the graph. Graphs with face-symmetric drawings are automatically
\emph{partial cubes} (or, equivalently, media~\cite{EppFalOvc-07}),
graphs in which the vertices may be labeled by bitvectors in such a
way that graph distance equals Hamming distance;
see~\cite{EppFalOvc-07} for the many applications of graphs of this
type.

\begin{figure}[t]
\centering
\includegraphics[width=3in]{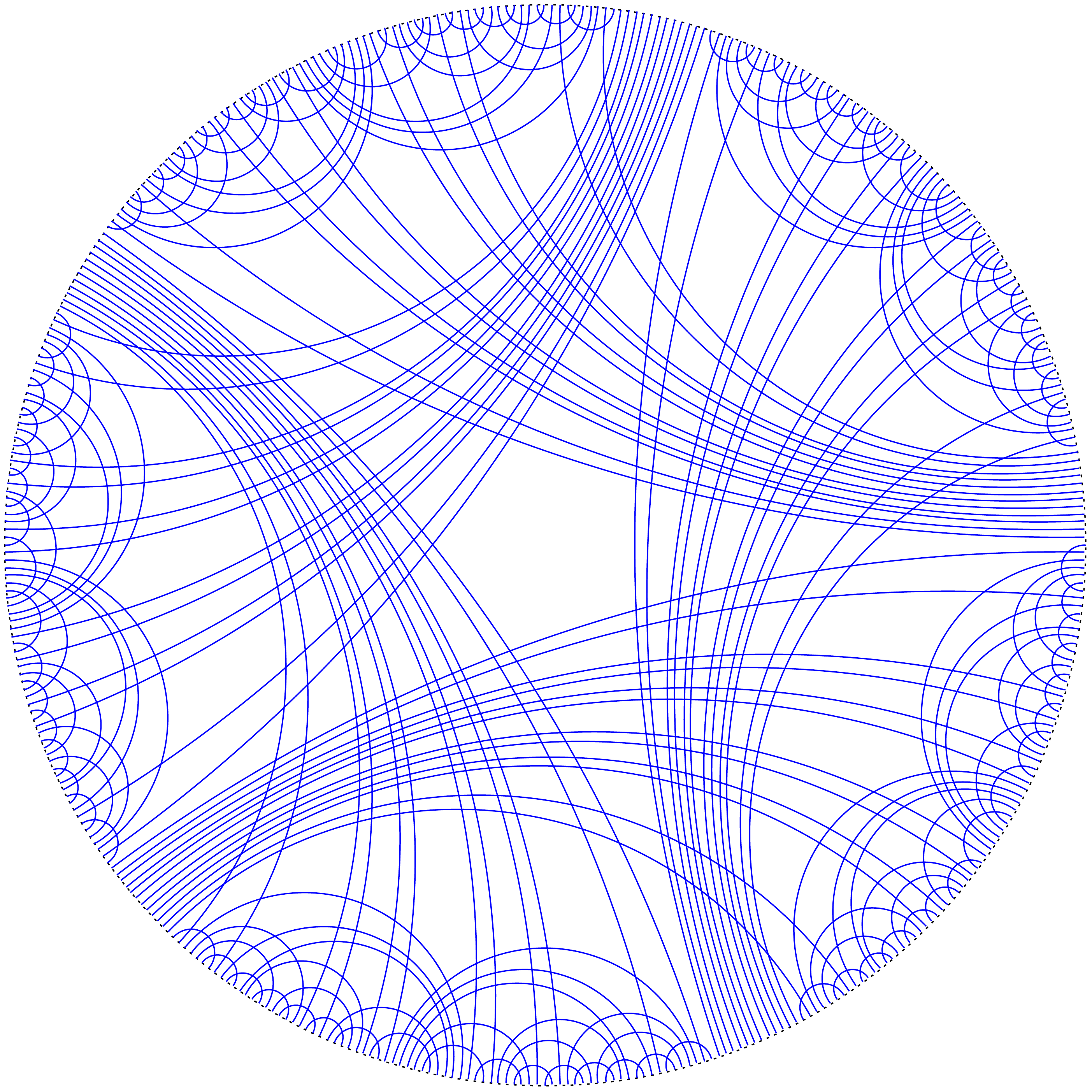}
\qquad
\includegraphics[width=3in]{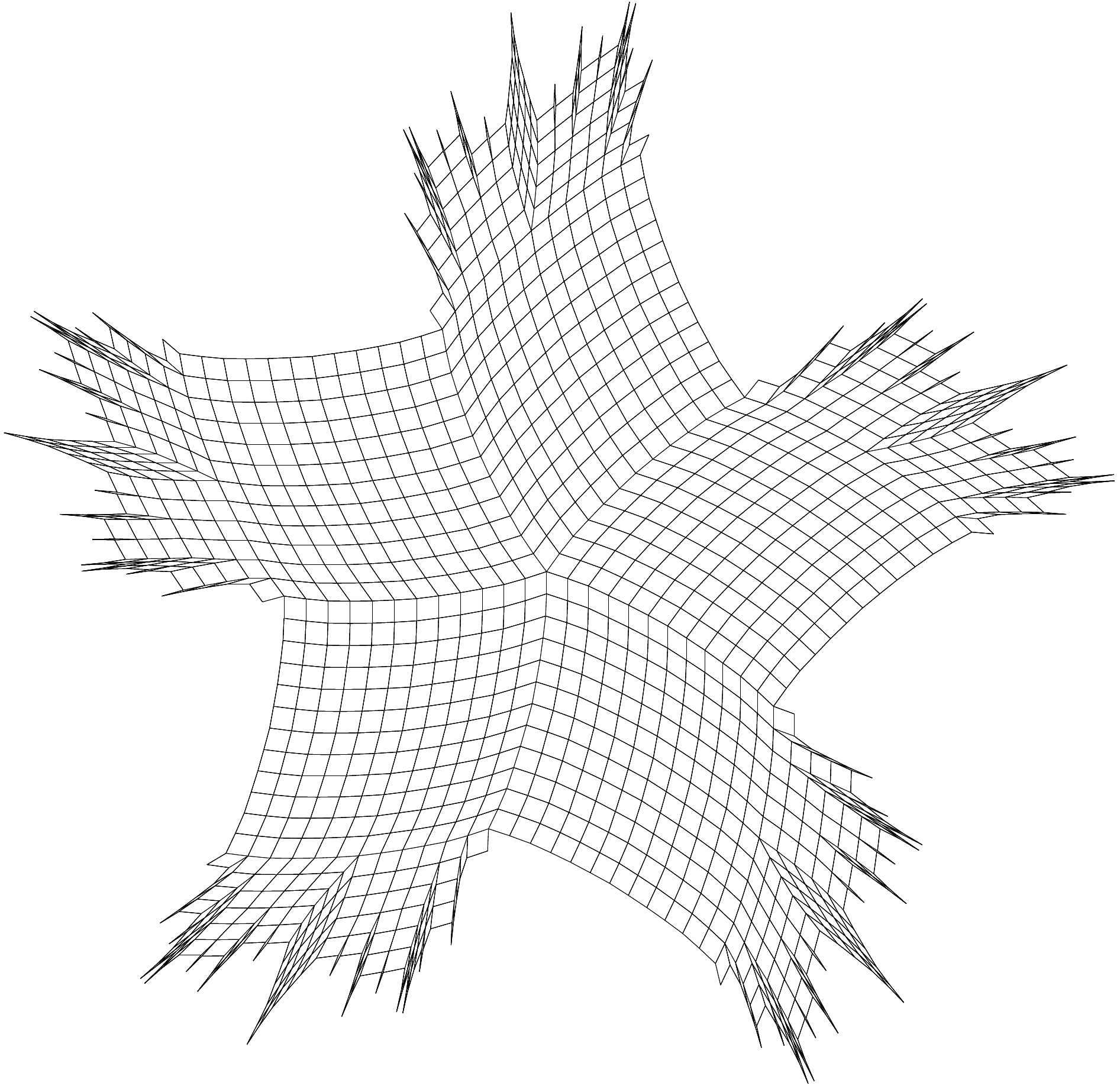}
\caption{A hyperbolic line arrangement (left), and its dual
squaregraph (right), drawn by our previous un-optimized algorithm with
all faces rhombi.}
\label{fig:old-220}
\end{figure}

Figure~\ref{fig:old-220} depicts an example, a hyperbolic line
arrangement with no three mutually intersecting lines, the
intersection graph of which requires five colors, as constructed by
Ageev~\cite{Age-DM-96}, and the planar dual graph of the arrangement.
The dual of a hyperbolic line arrangement with no three mutually
intersecting lines, such as this one, is a \emph{squaregraph}, a type
of planar median graph in which each internal face is a quadrilateral
and each internal vertex is surrounded by four or more
faces~\cite{CheDraVax-SODA-02, bendelt_chepoi_eppstein}. Any squaregraph may be
drawn in such a way that its faces are all rhombi, by our previous
algorithm~\cite{Epp-GD-04}, and the drawing produced in this way is
shown in the figure. However, note that some rhombi have such sharp
angles that they appear only as line segments in the figure, making
them difficult to view. Other rhombi, even those near the edge of the
figure, are drawn with overly wide angles, making them very legible
but detracting from the legibility of other nearby rhombi. Thus, we
are led to the problem of spreading out the angles more uniformly
across the drawing, in such a way as to optimize its angular
resolution, while preserving the property that all faces are rhombi.

As we show, this problem of optimizing the angular resolution of a
face-symmetric drawing may be solved in polynomial time, by
translating it into a problem of finding parametric shortest paths in
an auxiliary network. In the parametric shortest path problem, each
edge of a network is given a length that is a linear function of a
parameter $\lambda$; substituting different values of $\lambda$ into
these functions gives different real weights on the edges and
therefore different shortest path problems~\cite{Gus-JACM-83}. The
number of different shortest paths formed in this way, as $\lambda$ varies, may
be superpolynomial~\cite{Car-84}, but fortunately for our problem we
do not need to find all shortest paths for all values of $\lambda$. Rather, the
optimal angular resolution we seek may be determined as the maximum
value of $\lambda$ for which the associated network has no negative
cycles, and the drawing itself may be constructed using distances from
the source in this network at the critical value of $\lambda$. The
linear functions of our auxiliary network have a special
structure---each is either constant or a constant minus
$\lambda$---that allows us to apply an algorithm of Karp and
Orlin~\cite{KarOrl-DAM-81} for solving this variant of the parametric
shortest path problem, and thereby to optimize in polynomial time the
angular resolution of our drawings.

\section{Drawings with symmetric faces}
\label{section:drawings}

The algorithm of~\cite{Epp-GD-04} can be used to find a non-optimal face-symmetric
drawing, when one exists, in linear time; therefore, in the algorithms
here we will assume that such a drawing has already been given. We
summarize here the relevant properties of the drawings produced in
this way.

In a face-symmetric drawing, any two opposite edges of any face must
be parallel and of equal length. The transitive closure of this
relation of being opposite on the same face partitions the edges of
the drawing into equivalence classes, which we call \emph{zones}; any
zone consists of a set of edges that all have equal angles and equal
lengths. (Note, however, that edges in different zones may also have
equal angles and equal lengths.) The collection of line segments
connecting opposite pairs of edge midpoints within each face forms a
collection of curves which can be extended to infinity to form a
\emph{weak pseudoline arrangement}; each zone is formed by the drawing
edges that cross one of the curves of this arrangement. This
construction is depicted in Figure~\ref{fig:wpla}.

\begin{figure}[t]
\centering\includegraphics[scale=.6]{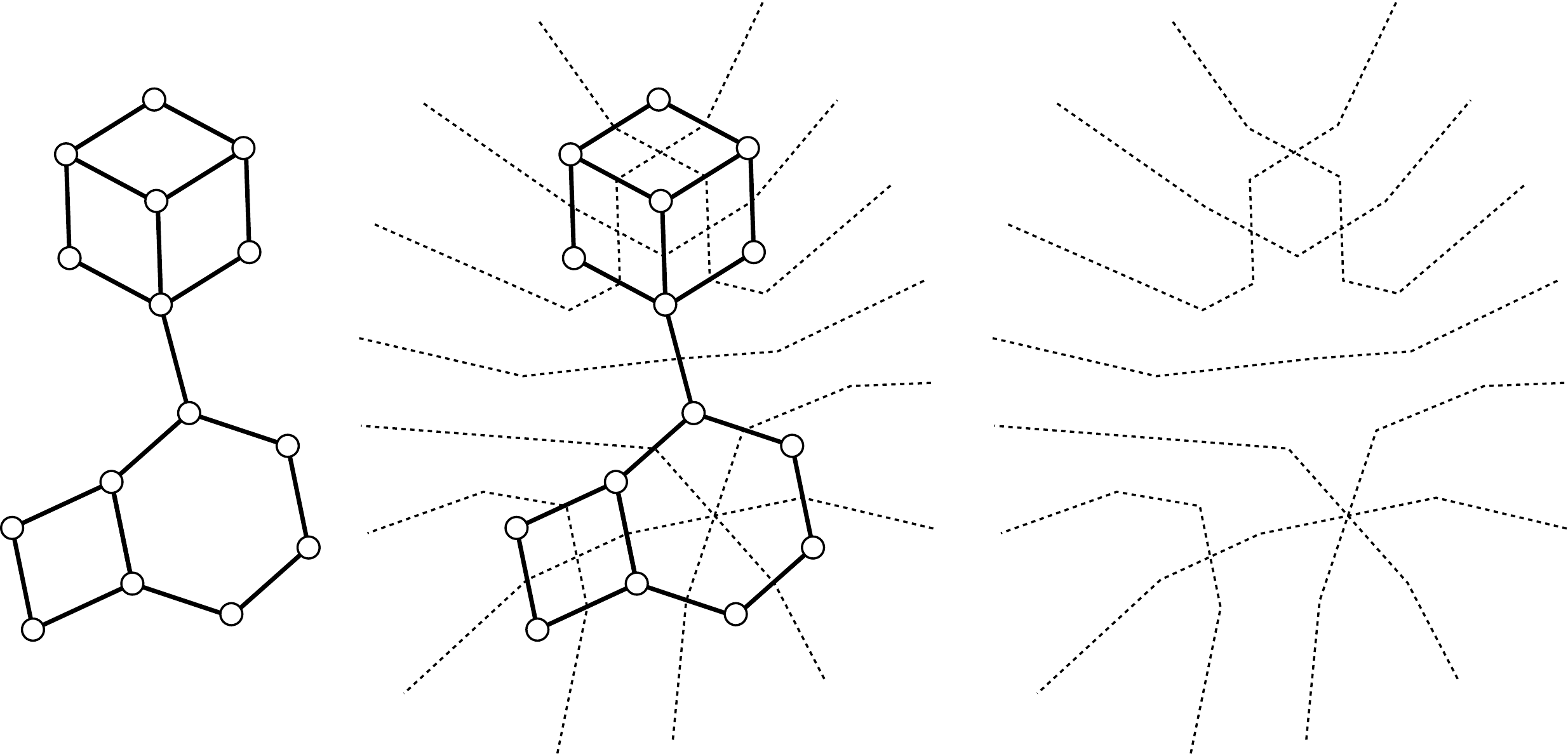}
\caption{Converting a face-symmetric drawing into a weak pseudoline
arrangement, by drawing line segments connecting opposite pairs of
edge midpoints within each face. From~\cite{Epp-GD-04}.}
\label{fig:wpla}
\end{figure}

Thus, the positions of all the vertices of the drawing are determined,
up to translation of the whole drawing, by a choice of a vector for each zone that specifies the direction and length of the zone's edges. For, if an arbitrarily chosen base vertex has its
placement fixed at the origin, the position of any other vertex $v$
must be the sum of the vectors corresponding to the weak pseudolines
that separate $v$ from the base vertex. All vertex positions may be
computed by performing a depth first search of the graph, setting the
position of each newly visited vertex $v$ to be the position of $v$'s
parent in the depth first search tree plus the vector for the zone
containing the edge connecting $v$ to its parent.

The algorithm from~\cite{Epp-GD-04} determines the vectors for each
zone as follows. Associate a unit vector with each end of each of the
pseudolines dual to the drawing, in such a way that these unit vectors
are equally spaced around the unit circle in the same cyclic order as
the order in which the pseudoline ends extend to infinity. The vector
associated with each zone is then simply the difference of the two
unit vectors associated with the corresponding pseudoline's ends,
normalized so that it is itself a unit vector.

As shown in \cite{Epp-GD-04}, this choice of a vector for each zone
leads to a face-symmetric drawing without crossings. Additionally, it
has two more properties, both of which are important and will be preserved by our optimization algorithm. First, if the planar embedding
chosen for the given graph has any symmetries, they will be reflected
in the dual pseudoline arrangement, in the choice of zone vectors, and
therefore in the resulting drawing.

Second, although it may not be possible to draw the given graph in
such a way that the outer face is convex, its concavities are all
mild. This can be measured by defining the \emph{winding number} of
any point $p$ on the boundary of the drawing, with respect to any
other point $q$ also on the boundary, as the sum of the turning angles
between consecutive edges along a path counterclockwise around the
boundary from $q$ to $p$. For any simple polygonal boundary, the
winding number from $p$ to $q$ is $2\pi$ minus the turning angle from
$q$ to $p$. In a convex polygon, all winding numbers would lie in the
range $[0,2\pi]$; in the drawings produced by this method, they
instead lie in the range $[-\pi,3\pi]$. That is, intuitively, the
sides of any concavity may be parallel but may not turn back towards
each other. It follows from this property that the vectors of each
zone may be scaled independently of each other, preserving only their
relative angles, and the resulting drawing will remain
planar~\cite{Epp-GD-04}. In particular, the step in the algorithm
of~\cite{Epp-GD-04} in which the zone vectors are normalized to unit
length leads to a planar drawing. A similar constraint on winding
numbers was used in the algorithm of~\cite{EppCar-GD-06} for finding
the optimal angular resolution of a tree drawing, and again led to the
ability to adjust edge lengths arbitrarily while preserving the
planarity of the drawing. Indeed, tree drawings may be seen as a very
special case of face-symmetric drawings in which there are no internal
faces to be symmetric.

Thus, we may formalize the problem to be solved, as follows: given a
face-symmetric drawing of a graph $G$, described as a partition of the
edges of $G$ into zones and a zone vector for each zone, we wish to
find a new set of zone vectors to use for a new drawing of $G$,
preserving the relative orientation of any two edges that meet at a
vertex of $G$, maintaining the constraint that all winding numbers lie
in the range $[-\pi,3\pi]$, and maximizing the angular resolution of
the resulting drawing.

\section{Algorithmic results}
\label{section:algorithm}

In this section we describe a polynomial time algorithm for the problem at hand.  As described above, the core of our algorithm is a routine that maps a planar face-symmetric drawing $G$ to an auxiliary graph $A$, such that the auxiliary graph may be used as input to a parametric shortest paths algorithm of Karp and Orlin \cite{KarOrl-DAM-81}, and the output of that algorithm describes a drawing $G'$ isomorphic to $G$ and with maximal angular resolution.  The algorithm of \cite{KarOrl-DAM-81} runs in polynomial time, so what remains to be seen is the correctness and running time of our mapping routine.

Table \ref{table:relationships} summarizes the intuitive relationships between concepts in the drawings $G$ and $G'$ and their representation in the auxiliary graph $A$.  As shown in the table, edge zones correspond to auxiliary graph vertices; angular resolution corresponds to the parametric variable $\lambda$; the change in angle of zones between $G$ and $G'$ corresponds to lengths of shortest paths in $A$; and constraints on the angles are represented by edges in $A$.  Figure \ref{figure:example_input} shows a small example input graph $G$ with five zones, and Figure \ref{figure:example_aux} shows the corresponding auxiliary graph $A$.

\begin{table}
\centering
\begin{tabular}{|p{3in}|p{3in}|}
\hline
\bf{Concept in drawing} & \bf{Representation in auxiliary graph} \\ \hline \hline
Zone $z_i$ & Vertex $v_i$ \\ \hline
Angular resolution & Parametric variable $\lambda$ \\ \hline
$\alpha$ is feasible angular resolution & No negative cycles when $\lambda \leq \alpha$ \\ \hline
Difference between angles of vectors for zone $z_i$ in unoptimized and optimized drawings & Shortest path distance $d(v_i)$ from $s$ to $v_i$ \\ \hline
Angle between $z_i$ and $z_j$ is $\geq \lambda$ & Edge with weight $\theta_G(z_i)-\theta_G(z_j)-\lambda$ \\ \hline
Interior faces are convex & Edge with weight $\pi + \theta_G(z_i)-\theta_G(z_j)$ \\ \hline
Exterior boundary is mildly convex & Opposing edges with weight  $w(v_i,v_j)= 3\pi + \theta_G(z_i) - \theta_G(z_j)$ and $w(v_j,v_i)=\pi+\theta_G(z_j)-\theta_G(z_i)$ \\ \hline
\end{tabular}
\caption{Relationships between drawing concepts and auxiliary graph components.}
\label{table:relationships}
\end{table}

\begin{figure}[t]
\centering
\vspace{.5in}
\includegraphics[width=4in]{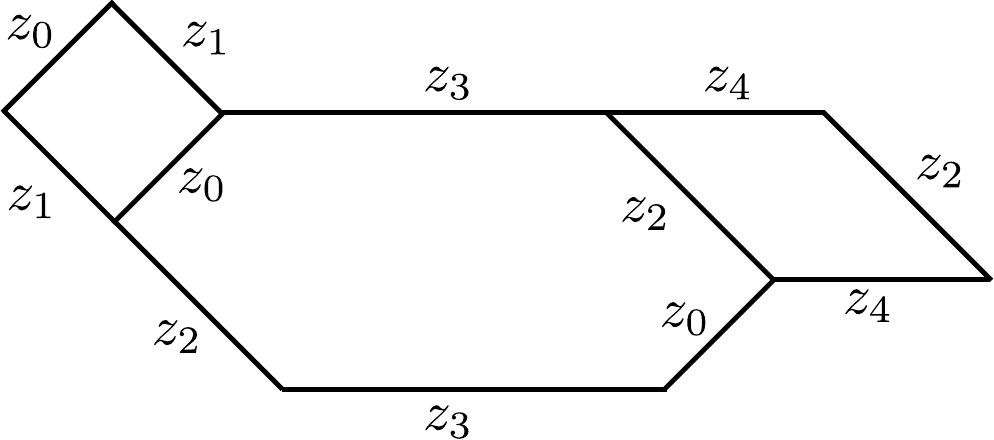} \hspace{.5in} \includegraphics[width=1.5in]{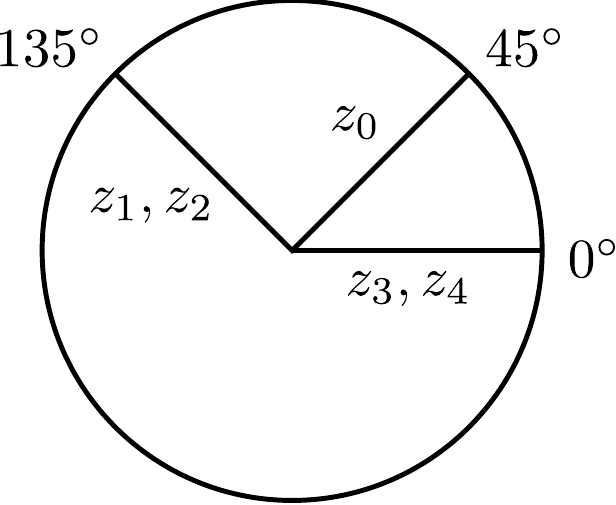}
\caption{Example input graph $G$, using degrees rather than radians.}
\label{figure:example_input}
\end{figure}

\begin{figure}[t]
\centering
\includegraphics[width=6in]{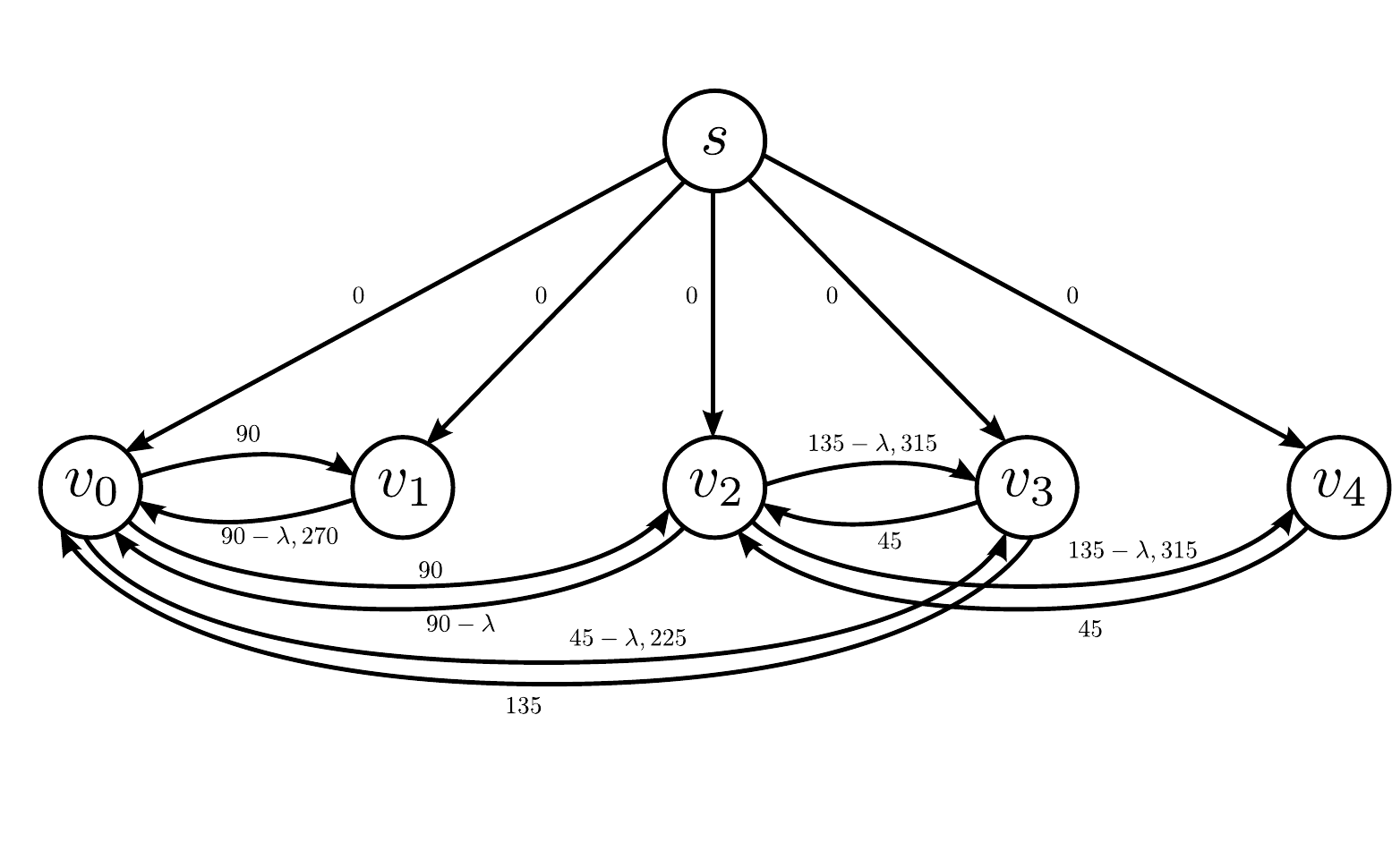}
\caption{Auxiliary graph $A$ for the input shown in Figure \ref{figure:example_input}, again using degrees.  Only the lowest-weight edge between any pair of vertices is shown.}
\label{figure:example_aux}
\end{figure}

We now define our notation formally and show that our reduction is correct.

\begin{definition}
Let
\begin{itemize}
\item $G$ be the planar face-symmetric graph and its drawing given as input;
\item $Z=\{z_0,z_1,\ldots\}$ be the zones of $G$;
\item $t=|Z|$ be the number of zones in $G$;
\item $A$ be a weighted, directed, auxiliary graph, such that every zone $z_i \in Z$ corresponds to a vertex $v_i$ in $A$, and every edge of $a$ has weight $w(v_i,v_j) = b_{v,w} - m_{v,w} \cdot \lambda$, where $b \in \mathbb{R}$ and $m_{v,w} \in \{0,1\}$ are per-edge constants, and $\lambda$ is a graph-wide variable;
\item $s$ be a special start vertex in $A$;
\item $A$ have edges $(v_i,v_j)$ and $(v_j,v_i)$ for every pair of vertices $v_i,v_j$;
\item $\lambda^*$ be the largest value of $\lambda$ such that $A$ contains no negative cycles;
\item $d(v_i)$ be the weight of the shortest path from $s$ to $v_i$ in $A$ when $\lambda=\lambda^*$;
\item $\theta_G(z_i)$ be the angle assigned to edges of zone $z_i$ in $G$, in radians counterclockwise from the $x$-axis;
\item $\angle_{G}(i,j) = \theta_G(z_j)-\theta_G(z_i)$ be the counterclockwise turning angle from $\theta_G(z_i)$ to $\theta_G(z_j)$ in $G$;
\item for every pair of distinct zones $z_i, z_j$, such that $\theta_G(z_i) \leq \theta_G(z_j)$ and there exists a face in $G$ with some edge from $z_i$ and some edge from $z_j$, $A$ contains the following edges:
   \begin{itemize}
   \item an edge $(v_j,v_i)$ with weight $w(v_j,v_i)=\theta_G(z_j)-\theta_G(z_i)-\lambda$,
   \item if a corresponding angle in $G$ is interior, an edge $(v_i,v_j)$ with weights $w(v_i,v_j)=\pi+\theta_G(z_i)-\theta_G(z_j)$, and
   \item if a corresponding angle in $G$ is exterior, opposing edges with weights $w(v_i,v_j)= 3\pi + \theta_G(z_i) - \theta_G(z_j)$ and $w(v_j,v_i)=\pi+\theta_G(z_j)-\theta_G(z_i)$;
   \end{itemize}
\item $G'$ be the output drawing of $G$;
\item $\theta_{G'}(z_i) = d(v_i)+\theta_{G}(z_i)$ be the angle assigned to edges in zone $z_i$ in the output drawing $G'$, in radians counterclockwise from the $x$-axis;
\item and $\angle_{G'}(i,j)=\theta_{G'}(z_j)-\theta_{G'}(z_i)$ be the angle between $z_i$ and $z_j$ in $G'$, analogous to $\angle_G(i,j)$ in $G$.
\end{itemize}
\end{definition}

\begin{lemma}
\label{lemma:path_lengths}
If $A$ contains any edge from $v_i$ to $v_j$ with weight $w(v_i,v_j)=W$, then $d(v_j) \leq d(v_i)+W$.
\end{lemma}

\begin{proof}
The shortest path to $v_i$ in $A$, followed by the edge $(v_i,v_j)$, is a path to $v_j$ with total weight $d(v_i)+W$.  Since $d$ is defined for $\lambda=\lambda^*$,  $A$ has no negative cycle; so the shortest path to $v_j$ has weight $d(v_j) \leq d(v_i)+W$.
\end{proof}

\begin{lemma}
\label{lemma:angles}
$G'$ has angular resolution $\lambda^*$, and every interior face of $G'$ is non-concave.
\end{lemma}

\begin{proof}
Let $e_i$ and $e_j$ be any pair of edges in $G$ that form an angle on some interior or exterior face.  If $e_i$ has the lesser absolute angle measure, then by construction there exists an edge $(v_j,v_i)$ with weight $w(v_j,v_i)=\theta_G(z_j)-\theta_G(z_i)-\lambda^*$, where $v_i$ and $v_j$ are the zones for $e_i$ and $e_j$, respectively.  Thus by Lemma \ref{lemma:path_lengths},
\begin{eqnarray*}
d(v_i) &\leq& d(v_j)+w(v_j,v_i) \\
d(v_i)-d(v_j) &\leq& (\theta_G(z_j)-\theta_G(z_i)-\lambda^*) \\
-d(v_i)+d(v_j) &\geq& -\theta_G(z_j)+\theta_G(z_i)+\lambda^* \\
d(v_j)+\theta_G(z_j)-d(v_i)-\theta_G(z_i) &\geq&\lambda^* \\
(\theta_{G'}(z_j))-(\theta_{G'}(z_i)) &\geq& \lambda^* \\
\angle_{G'}(i,j) &\geq& \lambda^* ;
\end{eqnarray*}
\noindent so every corner angle in $G'$ has measure no less than $\lambda^*$.

Now suppose $\angle e_i e_j$ is interior.  Then by construction there also exists an opposing edge $(v_i,v_j)$ with weight $w(v_i,v_j)=\pi+\theta_G(z_i)-\theta_G(z_j)$.  So again by Lemma \ref{lemma:path_lengths},
\begin{eqnarray*}
d(v_j) &\leq& d(v_i)+(\pi+\theta_G(z_i)-\theta_G(z_j)) \\
d(v_j)+\theta_G(z_j)-d(v_i)-\theta_G(z_i) &\leq& \pi \\
(\theta_{G'}(z_j))-(\theta_{G'}(z_i)) &\leq& \pi \\
\angle_{G'}(i,j) &\leq& \pi ,
\end{eqnarray*}
\noindent which implies that $G'$ contains no concave face.
\end{proof}

\begin{lemma}
\label{lemma:winding_numbers}
The winding number of any $p$ and $q$ on the boundary of $G'$ is in the range $[-\pi,3\pi]$.
\end{lemma}

\begin{proof}
As in Lemma \ref{lemma:angles}, let $e_i$ and $e_j$ be edges in $G$ such that $\theta_G(i) \leq \theta_G(j)$.  If $e_i$ and $e_j$ form an angle on the boundary of $G$, then the corresponding vertices $v_i$ and $v_j$ in $A$ have opposing edges with weights $w(v_i,v_j)= 3\pi + \theta_G(z_i) - \theta_G(z_j)$ and $w(v_j,v_i)=\pi+\theta_G(z_j)-\theta_G(z_i)$.  Then by algebraic manipulations symmetric to those in Lemma \ref{lemma:angles},
\[ -\pi \enspace \leq \enspace \angle_{G'}(i,j) \enspace \leq \enspace 3\pi . \]
\end{proof}

\begin{theorem}
The output graph $G'$ is a planar face symmetric drawing isomorphic to the input drawing $G$, such that all boundary winding numbers lie in the range $[-\pi,3\pi]$, and the angular resolution of $G'$ is maximal for any such drawing.  Given $G$, $G'$ may be generated in $O(t^3)$ time.
\end{theorem}

\begin{proof}
Identifying the zones of $G$ and constructing $A$ may be done naively in $O(t^2)$ time.  The value $\lambda^*$ and corresponding distances $d(v_i)$ for every vertex $v_i$ in $A$ may be reported by invoking the algorithm of Karp and Orlin.  This algorithm runs in $O(n^3)$ time, where $n$ is the number of vertices in the input; our graph $A$ has $t$ vertices, so this step takes $O(t^3)$ time.  The $d(v_i)$ distances define the zone angles $\theta_{G'}$, allowing the drawing $G'$ to be output in linear time as described in Section \ref{section:drawings}.

Our method of choosing new angles for the zone vectors implies that $G'$ is isomorphic to $G$ and every interior face of $G'$ is centrally symmetric.  By Lemmas \ref{lemma:angles} and \ref{lemma:winding_numbers}, every interior face of $G'$ is convex, and the exterior boundary of $G'$ satisfies the convexity constraint.

Finally we must show that $G'$ has maximal angular resolution.  Let $H$ be any correct drawing of the graph $G$, and let $H$ have angular resolution $\lambda_H$.  Then $A$ has no negative cycles for $\lambda=\lambda_H$: for, suppose we replace every weight $w(v_i,v_j)$ with $w'(v_i,v_j)=w(v_i,v_j)+\theta_H(v_i)-\theta_H(v_j)$, where $\theta_H(v_k)$ is the angle assigned to zone $z_k$ in $H$.  Any cycle contains an equal number of $+\theta_H(v_k)$ and $-\theta_H(v_i)$ terms for each $v_i$ in the cycle, so our transformation does not change total cycle weights, and hence preserves negative cycles.   Each edge has nonnegative weight, so no negative cycles exist in $A$.  We use the Karp-Orlin algorithm to find the largest value $\lambda^*$ for which $A$ contains no negative cycle, and by Lemma \ref{lemma:angles}, the output graph $G'$ has angular resolution $\lambda^*$.  Thus the angular resolution of $G'$ is greater than or equal to that of any correct drawing $H$.
\end{proof}

\section{Experimental results}

We implemented a simplified version of the algorithm described in Section \ref{section:algorithm} in the Python language.  Our simpler algorithm performs a numerical binary search for $\lambda^*$ rather than an implementation of the Karp-Orlin parametric search algorithm.  We make binary search decisions by generating the auxiliary graph as described in Section \ref{section:algorithm}, substituting in the value of $\lambda$ to obtain real-valued edge weights, then checking for negative cycles with a conventional Bellman-Ford shortest paths computation.  This algorithm runs in $O(t^3 \log W)$ time, where $t$ is the number of pseudolines (zones) in the drawing and $W$ is the number of bits of numerical precision used.  Wall clock time ranged from a matter of seconds for small $t$ to roughly two minutes for $t=220$ and $W=64$.

Figure \ref{figure:cd15} shows a sample of input and output for our implementation.  The input is a correct, though suboptimal, planar face-symmetric drawing of a graph with 15 pseudolines generated by the algorithm presented in ~\cite{Epp-GD-04}.  The angular resolution of the output is visibly improved.  Figure \ref{figure:cd220} shows the output of the algorithm when the 220-pseudoline graph shown in Figure \ref{fig:old-220} is used as input.  That unoptimized drawing has angular resolution $2\pi/110 \approx .0571$ radians.  The left output drawing was produced by our optimization algorithm using all correctness constraints, and has angular resolution $2\pi/75\approx .0838$ radians.  The right drawing is the result of removing the constraint that concavities have opening angles of at least $\pi$ radians, which may in general result in a nonplanar drawing, but in this case yields a planar drawing with an even greater angular resolution of $2\pi/30\approx .209$ radians.

\begin{figure}[t]
\centering
\includegraphics[width=3in]{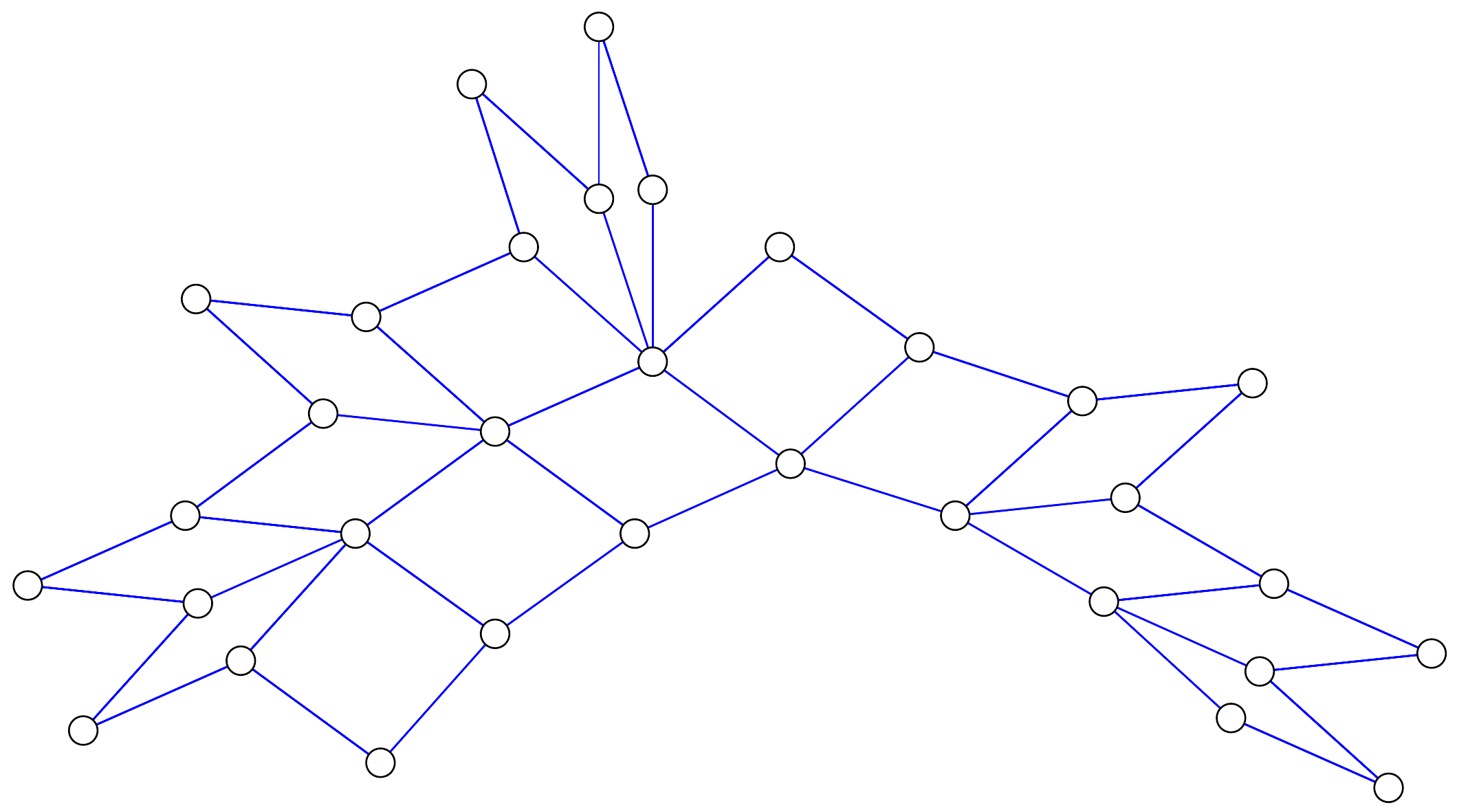} \includegraphics[width=3in]{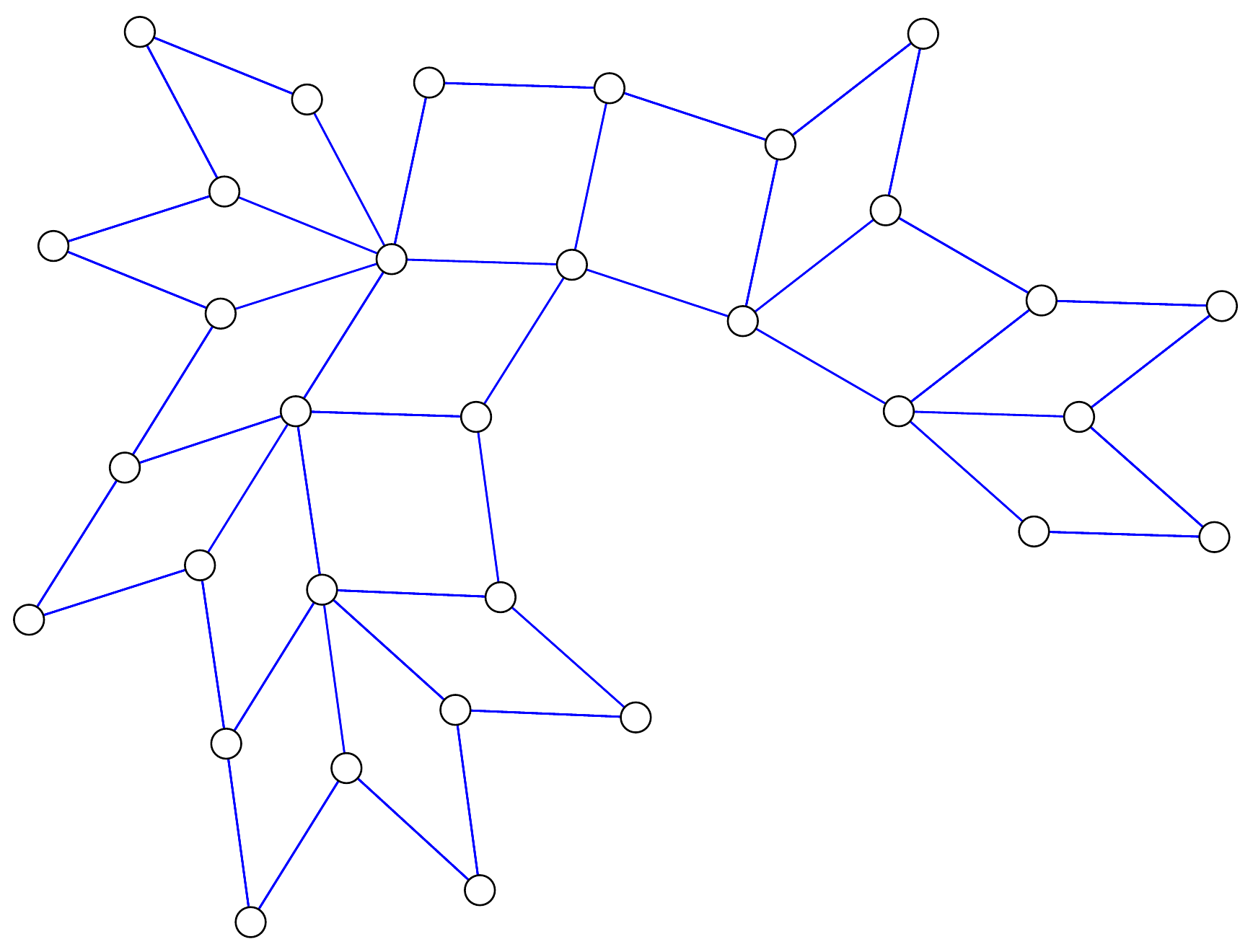}
\caption{Example input (left) and output (right) with 15 pseudolines.}
\label{figure:cd15}
\end{figure}

\begin{figure}[t]
\centering
\includegraphics[width=3.1in]{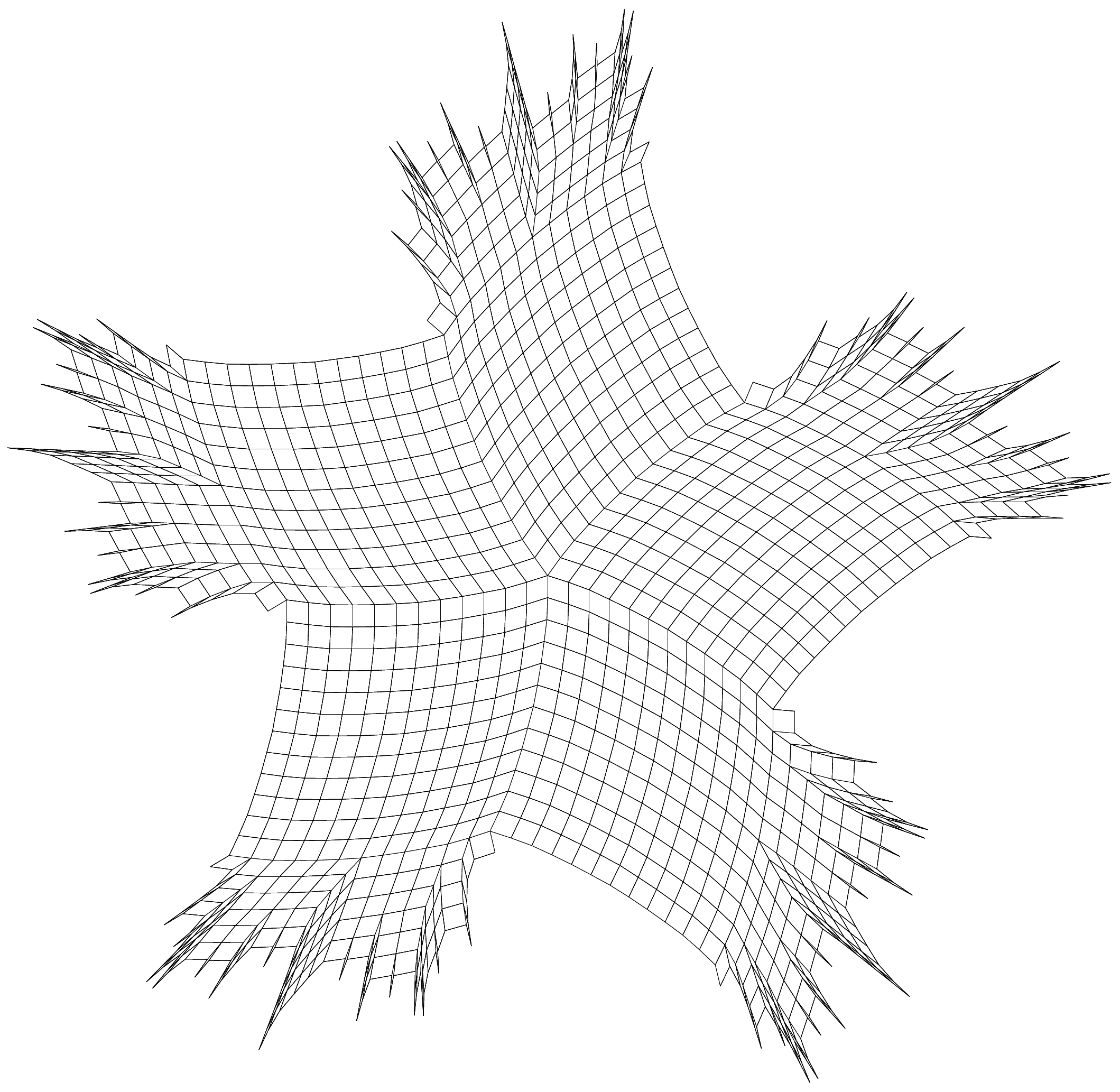} \includegraphics[width=3.1in]{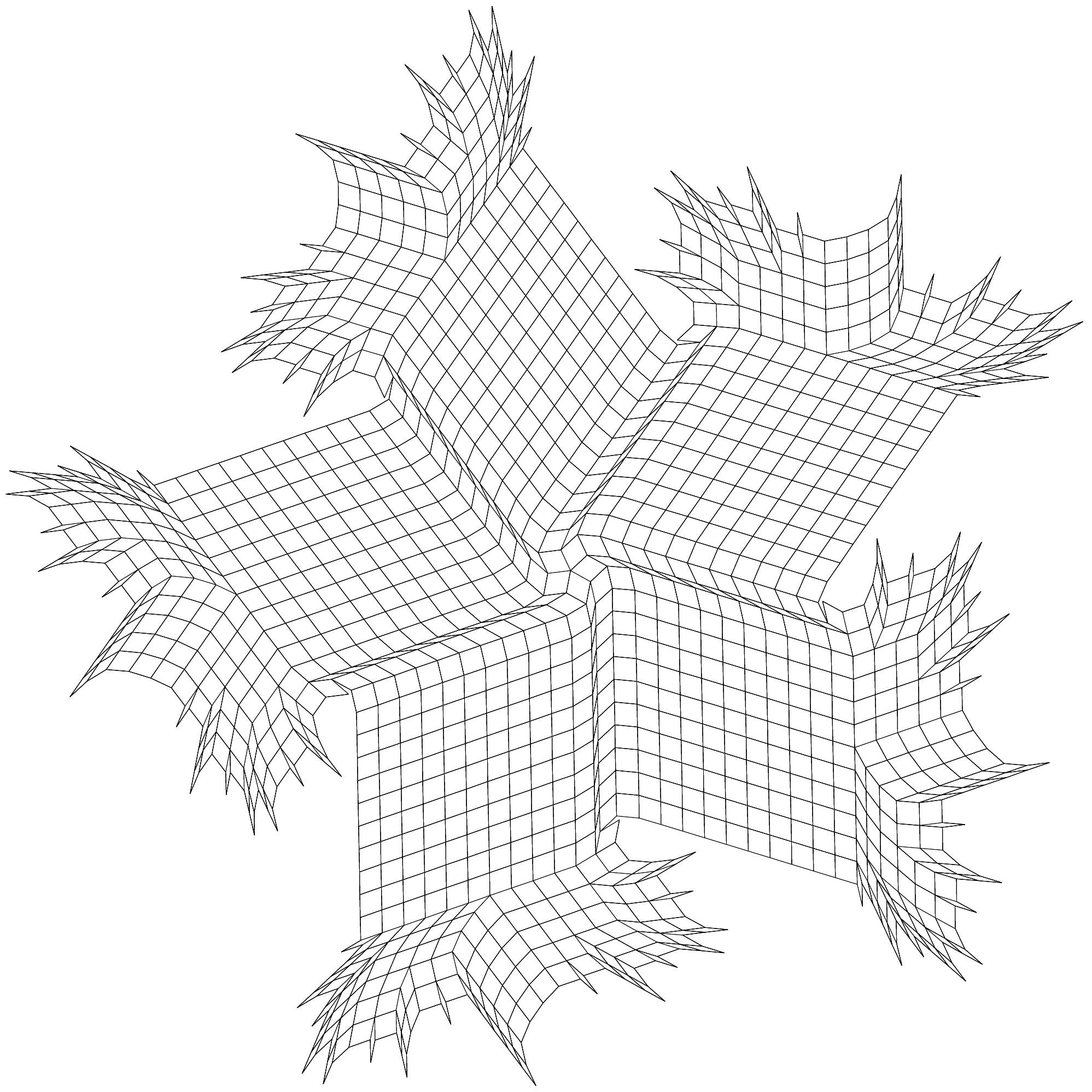}
\caption{Output of our implementation for the graph shown in Figure \ref{fig:old-220}, safely optimized (left) and unsafely optimized (right).}
\label{figure:cd220}
\end{figure}

\section{Conclusion}

We have described two algorithms for generating face-symmetric drawings with optimal angular resolution.  The first algorithm runs in strictly cubic time and uses a subsidiary parametric shortest paths algorithm as a black-box subroutine.  The second algorithm runs in pseudopolynomial time and relies on the simpler Bellman-Ford shortest paths algorithm.  We implemented the latter algorithm and found that it generates output that is both numerically and visibly improved over unoptimized drawings.

Finally we offer the following possible directions for future research:
\begin{itemize}
\item We choose the edge length for each zone arbitrarily; can choosing them more carefully lead to improved legibility?
\item Can this approach be applied to related types of drawings, such as the projections of high dimensional grid embeddings also studied in \cite{Epp-GD-04}?
\item Our algorithm respects a fixed embedding.  What about allowing for different embeddings, e.g. flipping parts of the graph at articulation vertices?  Is it still possible to optimize angular resolution efficiently in this more general problem?
\item What about other angle optimization criteria, such as those defined by all angles in the drawing, rather than merely the sharpest angle?
\end{itemize}

\bibliography{oarpdsf}
\bibliographystyle{abbrv}

\end{document}